\documentclass[conference]{IEEEtran}
\IEEEoverridecommandlockouts
\usepackage{cite}
\usepackage{amsmath,amssymb,amsfonts,amsthm}
\usepackage{algorithmic}
\usepackage{graphicx}
\usepackage{textcomp}
\usepackage{subcaption}
\usepackage{xcolor}
\usepackage{soul}
\usepackage{multirow}
\usepackage{makecell}
\usepackage{booktabs,siunitx}
\def\BibTeX{{\rm B\kern-.05em{\sc i\kern-.025em b}\kern-.08em
    T\kern-.1667em\lower.7ex\hbox{E}\kern-.125emX}}

\newtheorem{thm}{Theorem}

\makeatletter
\newcommand{\linebreakand}{%
  \end{@IEEEauthorhalign}
  \hfill\mbox{}\par
  \mbox{}\hfill\begin{@IEEEauthorhalign}
}
\makeatother
    
\begin{document}

\title{Deep Pseudo-Proximal Map: A Self-Supervised Data-Fitting Agent for Iterative Reconstruction
}

\author{\IEEEauthorblockN{Haley Duba-Sullivan}
\IEEEauthorblockA{
\textit{Purdue University}\\
West Lafayette, IN, USA \\
hduba@purdue.edu}
\and
\IEEEauthorblockN{Emma J. Reid}
\IEEEauthorblockA{
\textit{Oak Ridge National Laboratory}\\
Oak Ridge, TN, USA \\
reidej@ornl.gov} 
\linebreakand 
\IEEEauthorblockN{Charles A. Bouman}
\IEEEauthorblockA{
\textit{Purdue University}\\
West Lafayette, IN, USA \\
bouman@purdue.edu}
\and
\IEEEauthorblockN{Gregery T. Buzzard}
\IEEEauthorblockA{
\textit{Purdue University}\\
West Lafayette, IN, USA \\
buzzard@purdue.edu}}
 
\maketitle

\begin{abstract}
Iterative reconstruction methods alternate a regularizer with the data-fitting proximal map. In many practical cases, this proximal map has no analytic solution and must be approximated by an inner loop solver at every iteration, compounding the cost of the outer reconstruction loop. To address this, we propose the pseudo-proximal map (PPM), a reformulation of the data-fitting proximal map as the MMSE estimate of a virtual model. We implement the deep PPM as a self-supervised network trained only on sampled Gaussian noise, requiring no ground-truth training images. The deep PPM can be trained for any operator for which the forward model and its transpose can be evaluated, with provable equivalence to the proximal map when $A$ is linear. We validate the deep PPM on Gaussian deblurring and super-resolution, where the proximal map has an analytic solution, and on X-ray computed tomography (XCT), where no analytic solution exists. Used as the data-fitting agent within an iterative reconstruction method, the deep PPM reproduces the reference reconstruction to within 1\% NRMSE for all three operators, and for XCT, it replaces the inner conjugate-gradient loop with a single network evaluation that is $18 \times$ faster.

\end{abstract}


\section{Introduction}

Reconstructing images from noisy and incomplete measurements is a central problem in computational imaging. 
Iterative reconstruction methods such as ADMM~\cite{boyd2011distributed}, plug-and-play (PnP)~\cite{venkatakrishnan2013plug,buzzard2018plug}, and multi-agent consensus equilibrium (MACE)~\cite{sridhar2020distributed} rely on repeatedly evaluating a data-fitting proximal map.
In many practical settings, this proximal map does not admit an analytic solution. 
For example, in sparse-view computed tomography (CT) the projection operator is large and geometry dependent~\cite{sidky2008}, while ptychography and holography involve forward models based on coherent wave propagation~\cite{rodenburg2008,javidi2021}.
In these settings, the proximal map is often approximated using iterative solvers such as conjugate gradient (CG)~\cite{hestenes1952methods} or gradient descent (GD)~\cite{nocedal2006numerical}, creating inner loops that compound computational cost.

Learning-based approaches have been proposed to reduce this computational cost.
Adaptive reconstruction networks~\cite{gossard2024,terris2025} incorporate knowledge of the forward operator, but require training using paired training data.
Learned prior proximal operators~\cite{fang2023} approximate the proximal map of an implicit regularizer, addressing the image prior rather than the data-fitting term.
The learned data-fitting solver proposed in~\cite{fan2017} approximates the matrix inversion arising in the ADMM data-fitting update with a learned network.
The deep proximal map introduced in~\cite{yang2020deepprox} uses a virtual model to replace the data-fitting proximal map with a network but trains on triples of noise, images, and residuals and requires a set of input images for training.

In this paper, we propose the pseudo-proximal map (PPM), a principled reformulation of the data-fitting proximal map as the minimum mean squared error (MMSE) estimate of the PnP virtual model that is provably equivalent when $A$ is linear.
With this reformulation, we train the deep PPM using only sampled Gaussian noise and a black-box implementation of the forward model $A$ and its transpose $A^T$, taking the form of a residual correction to the input rather than an approximate inverse as in previous work.
We validate the deep PPM on Gaussian deblurring and super-resolution (SR) where analytic solutions to the proximal map are available for direct comparison.
We also test the deep PPM on X-ray CT (XCT) reconstruction, which does not admit an analytic solution.
Additionally, we show that MACE reconstructions using the deep PPM closely match the reference reconstruction to within 1\% NRMSE for all three operators, while replacing the inner CG loop for XCT with a single network evaluation that is $18\times$ faster.

\section{Deep Pseudo-Proximal Map}

Inverse imaging problems require inverting a physical measurement system modeled as
\begin{equation} \label{eq:forward_model}
    y = A(x) + \sigma W,
\end{equation}
where $A: \mathbb{R}^N \to \mathbb{R}^M$ is the forward model, $\sigma > 0$ is the measurement noise standard deviation, and $W \sim \mathcal{N}(0, I)$ is additive white Gaussian noise (AWGN).
The data-fitting proximal map for this forward model is defined for $\gamma > 0$ as
\begin{equation}\label{eq:pm}
    F_{\sigma, \gamma}(v; y)
    = \arg\min_x \left\{ \frac{1}{2\sigma^2}\|y - A(x)\|^2
    + \frac{1}{2\gamma^2}\|x - v\|^2 \right\},
\end{equation}
where $v \in \mathbb{R}^N$ is any input image.

The key observation motivating our proposed PPM is that $F_{\sigma, \gamma}(v; y)$ can be interpreted as the maximum a posteriori (MAP) estimate under a Gaussian prior centered at $v$.
To see this, we apply the PnP virtual model~\cite[Ch.~10]{bouman_2022_foundations} to the data-fitting proximal map.
The virutal model is given by
\begin{equation}\label{eq:synth_sys}
    \tilde{X} = v + \gamma\tilde{\mathcal{E}}; \quad
    \tilde{Y} = A(\tilde{X}) + \sigma\tilde{W},
\end{equation}
where $\tilde{\mathcal{E}}, \tilde{W} \sim \mathcal{N}(0, I)$ are independent AWGN vectors.
Under this model, $F_{\sigma, \gamma}(v; y)$ is the MAP estimate of $\tilde{X}$ given $\tilde{Y}=y$, i.e.,
\begin{equation}
    F_{\sigma,\gamma}(v;y)
    = \arg\min_x \{ -\log p(x \mid \tilde{Y}=y)\}.
\end{equation}

Then, we define the PPM $\tilde{F}_{\sigma, \gamma}(v; y)$ as the corresponding conditional expectation of $\tilde{X}$ given $\tilde{Y}=y$, i.e., 
\begin{equation}\label{eq:cond-exp}
    \tilde{F}_{\sigma, \gamma}(v; y)
    = \mathbb{E}_{v,\gamma}[\tilde{X} \mid \tilde{Y} = y].
\end{equation}
When $A$ is linear, the posterior is Gaussian, so $F_{\sigma,\gamma}(v;y)$ and $ \tilde{F}_{\sigma, \gamma}(v; y)$ are equivalent, as formalized in Theorem~\ref{thm:equiv}.

\begin{thm}\label{thm:equiv}
Assume the virtual model from~\eqref{eq:synth_sys} and let $A$ be linear. Then the PPM is equivalent to the data-fitting proximal map, i.e.
\begin{equation}
F_{\sigma, \gamma}(v; y) = \tilde{F}_{\sigma, \gamma}(v; y) \text{ for all $v$ and $y$.}\end{equation}
\end{thm}

\begin{proof}
Under the virtual model defined in~\eqref{eq:synth_sys}, Bayes' rule gives
\begin{equation}
    -\log p_{v,\gamma}(x \mid \tilde{Y} = y)
    = \frac{1}{2\sigma^2}\|y-Ax\|^2
    + \frac{1}{2\gamma^2}\|x-v\|^2,
\end{equation}
up to an additive constant with respect to $x$.

Since $A$ is linear, this is quadratic in $x$ with positive definite Hessian $\frac{1}{\sigma^2}A^TA + \frac{1}{\gamma^2}I.$
So, the posterior $p(x\mid \tilde{Y} = y)$ is a nondegenerate Gaussian whose mean and mode coincide.
Therefore, since the PPM and the data-fitting proximal map are exactly the mean and mode of this posterior, 
\begin{equation}
    F_{\sigma,\gamma}(v;y)=\tilde{F}_{\sigma,\gamma}(v;y)
\end{equation}
for all $v$ and $y$.
\end{proof}

We emphasize that the Gaussian prior used in the virtual model of~\eqref{eq:synth_sys} is only used to provide a probabilistic interpretation of the data-fitting proximal map and to define the proposed PPM. 
It is not assumed to represent the true prior distribution of the original inverse problem. 
Thus, the proposed PPM can be incorporated into reconstruction algorithms with arbitrary image priors.

We can implement the PPM as a network trained to minimize the mean square error (MSE) loss, since the MMSE estimator is equivalent to the conditional expectation~\cite{chan2021textbook}.
We define the deep PPM $\tilde{F}_{\sigma, \gamma}^\theta(v; y)$ in residual form as
\begin{equation}\label{eq:ppm_net}
    \tilde{F}_{\sigma, \gamma}^\theta(v; y)
    = v + \tilde{H}_{\sigma, \gamma}^\theta(z),
\end{equation}
where $z = A^T(y - Av)$ is the back-projected residual and $\tilde{H}_{\sigma, \gamma}^\theta$ is a network with weights $\theta$. 
The weights are trained to minimize the MSE loss given by
\begin{equation}\label{eq:loss}
    \mathcal{L}(\theta) = \frac{1}{K}\sum_{k=0}^{K-1}
    \|\gamma\epsilon_k
    - \tilde{H}_{\sigma, \gamma}^\theta(z_k)\|^2,
\end{equation}
with $\epsilon_k, w_k$ sampled independently from $\mathcal{N}(0, I)$ and $z_k$ the corresponding back-projected residual.

When $A$ is linear, the residual satisfies
\begin{equation}
    y - Av = \gamma A\epsilon + \sigma w.
\end{equation} 
Importantly, this residual depends only on sampled Gaussian noise and the operator $A$.
Thus, by formulating the deep PPM in residual form, we remove any dependence of the trained network on the input image $v$.
In particular, $\tilde{H}_{\sigma, \gamma}^\theta$ is independent of the input $v$ and no ground-truth images are required for training.

We note that this independence is a consequence of linearity.
For nonlinear $A$, the residual $y - A(v)$ retains its dependence on $v$, so training requires either samples of $v$ from a representative distribution or a linearization of $A$ about $v$.
The PPM remains well-defined in that setting, but its training and verification are beyond the scope of this paper.

\section{Implementation Details}
\label{sec:setup}

\subsection{Forward Operators}

We evaluate the deep PPM on Gaussian deblurring and $4\times$ SR, both of which admit an analytic solution to the data-fitting proximal map that can be efficiently computed in the Fourier domain~\cite{chan2016plug}.
We call this the analytic proximal map, and use it as the reference proximal map.
For Gaussian deblurring, $A$ is convolution with a $7 \times 7$ Gaussian kernel of standard deviation $1.0$.
For $4\times$ SR, $A$ is convolution with a $13 \times 13$ Gaussian kernel of standard deviation $2.0$ followed by $4\times$ spatial subsampling.

We also evaluate the deep PPM on XCT reconstruction, which does not admit an analytic solution to the data-fitting proximal map.
In this case, we use a CG solution run to relative tolerance $10^{-6}$ as the reference proximal map.
For XCT reconstruction, $A$ is a parallel-beam projection operator with uniformly spaced angles between 0 and 180 degrees, implemented with the ASTRA toolbox~\cite{vanaarle2015astra}.

\subsection{Training Details}

We train one deep PPM network per operator, fixing $\sigma = 0.01$ and $\gamma = 0.05$ for all tasks.
The network $\tilde{H}^\theta_{\sigma,\gamma}$ is built on a DRUNet architecture~\cite{zhang2021drunet} that takes the back-projected residual $z$ as input and outputs the estimated correction $\gamma\epsilon$.
Training pairs $(z_k, \gamma\epsilon_k)$ are generated on the fly, with $\epsilon_k$ and $w_k$ drawn fresh at every step, so no image data of any kind enters training.
Each epoch consists of 10{,}000 training samples at the operator's native size ($256 \times 256$ for deblurring and SR, $128 \times 128$ for XCT), and the validation loss is computed on a fixed held-out set of 1{,}000 samples.
We minimize the MSE loss given in \eqref{eq:loss} with Adam~\cite{kingma2015adam} for 100 epochs at batch size 96, with an initial learning rate of $10^{-4}$, five epochs of linear warmup, and cosine annealing to $10^{-6}$.
We use the checkpoint with the lowest validation loss for all experiments.

\subsection{Testing Details}

For deblurring and SR, we use 50 crops of size $256 \times 256$ from the DIV2K validation set~\cite{agustsson2017div2k}.
For XCT, we use 28 slices from the LoDoPaB-CT dataset~\cite{leuschner2021lodopab}, downsampled to size $128\times 128$.
We emphasize that none of these images are used in training; each PPM is trained on sampled Gaussian noise alone.

We compare the deep PPM against CG and GD, each run for at most 50 iterations with relative convergence tolerance $10^{-6}$.
All three proximal map variants are evaluated first as standalone proximal maps and then as the data-fitting agent within MACE~\cite{buzzard2018plug, sridhar2020distributed} with 10 iterations using a BM3D denoiser~\cite{dabov2007bm3d} as the prior agent.
The proximal-map input $v$ and the initial guess for MACE are the blurry image for deblurring,  $4\times$ bicubic upsampled image for SR, and filtered backprojection (FBP) reconstruction for XCT.

For a proximal map variant $P$, we report the residual NRMSE with respect to the reference proximal map, i.e.
\begin{equation}\label{eq:res_nrmse}
    \text{Res-NRMSE} = \frac{\|(P(v;y) - v) - (F_{\sigma,\gamma}(v;y) - v)\|}{\|F_{\sigma,\gamma}(v;y) - v\|}.
\end{equation}
The Res-NRMSE measures the error in the correction rather than in the full output and is therefore not dominated by the shared input $v$.
For MACE we report the NRMSE of the reconstruction relative to the MACE reconstruction obtained with the reference proximal map, and relative to ground truth.
We additionally report mean wall-clock time per proximal map evaluation and per MACE run, measured on a single NVIDIA A100 80GB GPU. 

\section{Experimental Results}
\label{sec:results}
\begin{table}[t]
\centering
\caption{Quantitative comparison of each proximal map variant against the reference proximal map, both as a standalone map and as the data-fitting agent in MACE. The reference is the analytic proximal map for deblurring and SR and the CG solution for XCT. For Res-NRMSE / NRMSE, ``ref'' and ``GT'' denote error relative to the reference proximal map and to ground truth, respectively. Reported errors are averaged over 50 DIV2K crops (deblurring, SR) or 28 LoDoPaB-CT slices (XCT). Time for the proximal map is the mean per evaluation; time for MACE reconstruction is the mean wall-clock time of a full 10-iteration run. }
\label{tab:results}
\setlength{\tabcolsep}{4pt}
\footnotesize
\begin{tabular}{@{}c c r r r r r@{}}
\toprule
 & & \multicolumn{2}{c}{Proximal map} & \multicolumn{3}{c}{MACE reconstruction} \\
\cmidrule(lr){3-4} \cmidrule(lr){5-7}
Task & Method & \makecell{Res-NRMSE\\(\%; ref)} & \makecell{Time\\(ms)} & \makecell{NRMSE\\(\%; ref)} & \makecell{NRMSE\\(\%; GT)} & \makecell{Time\\(s)} \\
\midrule
\multirow{4}{*}{\makecell[c]{Deblur}}
 & Analytic & {--} & {--} & {--} & 5.63 & 21.6 \\
 & CG & 0.02 & 8.9 & 0.20 & 5.63 & 21.8 \\
 & GD & 0.05 & 40.0 & 0.21 & 5.63 & 22.1 \\
 & Deep PPM & 6.02 & 7.6 & 0.34 & 5.64 & 21.7 \\
\midrule
\multirow{4}{*}{\makecell[c]{4$\times$ SR}}
 & Analytic & {--} & {--} & {--} & 13.28 & 21.2 \\
 & CG & 0.01 & 3.7 & 0.03 & 13.28 & 21.2 \\
 & GD & 0.01 & 55.6 & 0.04 & 13.28 & 22.0 \\
 & Deep PPM & 5.88 & 7.5 & 0.14 & 13.29 & 21.2 \\
\midrule
\multirow{3}{*}{\makecell[c]{XCT}}
 & CG & {--} & 710.1 & {--} & 7.20 & 9.7 \\
 & GD & 0.10 & 6011.6 & 0.11 & 7.20 & 62.3 \\
 & Deep PPM & 11.77 & 38.7 & 0.70 & 7.29 & 4.3 \\
\bottomrule
\end{tabular}
\end{table}

Table~\ref{tab:results} compares the deep PPM to the reference proximal map, both as a standalone proximal map and as a data-fitting agent in MACE.
As a standalone proximal map, the deep PPM reproduces the correction $F_{\sigma,\gamma}(v;y) - v$ to a Res-NRMSE of approximately $6.0\%$ for both deblurring and SR.
For XCT, the Res-NRMSE with respect to CG is 11.8\%, about twice that of the other operators. 
Within MACE, the reconstruction obtained with the deep PPM differs from the reconstruction with the reference proximal map by less than 1\% on average for all three tasks.
Measured against ground truth, all proximal map variants reach the same reconstruction quality as the reference proximal map to within 0.1 percentage points.
Additionally, the cost of the deep PPM is a single network evaluation and does not depend on the conditioning of $A$.
For deblurring and SR this offers no advantage, since the analytic solution and CG each take only a few milliseconds; these tasks serve as verification benchmarks rather than speed benchmarks.
For XCT, however, the deep PPM replaces a 710\,ms CG solve (6.0\,s for GD) with a 38.7\,ms evaluation, an $18\times$ reduction in the cost of each data-fitting step.
Over a full MACE run the speedup is $2.3\times$ over CG and $14.5\times$ over GD.

\begin{figure*}[t]
    \centering\includegraphics[width=0.93\linewidth]{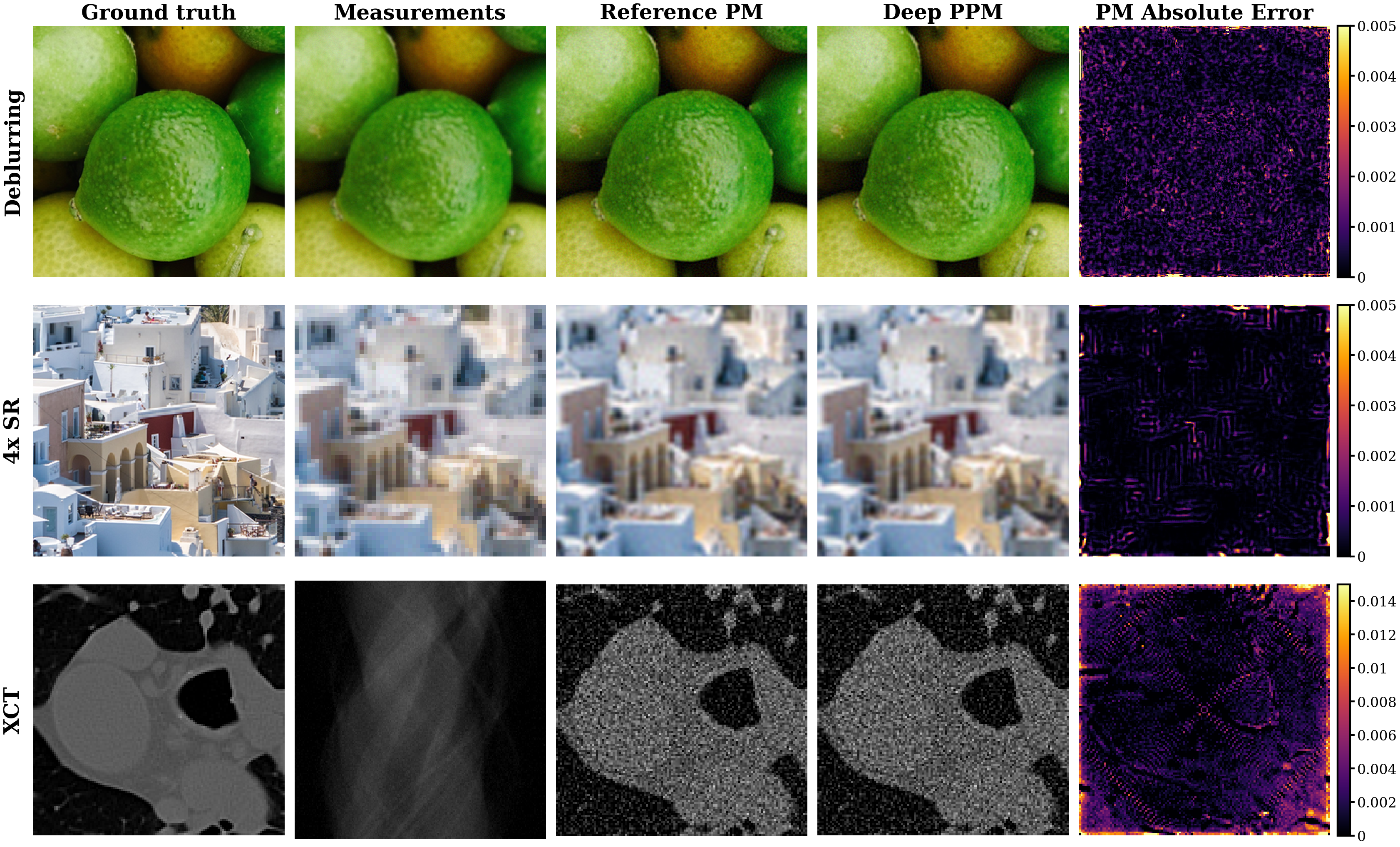}
    \caption{Comparison of deep PPM and reference proximal map for each operator. From left to right, each row displays the ground truth, measurement $y$, reference proximal map output $F_{\sigma,\gamma}(v;y)$, deep PPM output $\tilde{F}^\theta_{\sigma,\gamma}(v;y)$, and proximal map absolute error $|F_{\sigma,\gamma}(v;y) - \tilde{F}^\theta_{\sigma,\gamma}(v;y)|$ (mean over color channels). The reference is the analytic proximal map for deblurring and SR, and the CG solution for XCT. In all three cases the deep PPM output is visually indistinguishable from the reference.}
    \vspace{-5mm}
    \label{fig:prox}
\end{figure*}

Figure~\ref{fig:prox} compares outputs of the deep PPM with the reference proximal map for each operator.
In all three cases, the deep PPM output is visually indistinguishable from the reference proximal map, supported by the small absolute error maps.
For deblurring, the error is spatially unstructured noise.
For SR, the error concentrates on edges and periodic texture, which is exactly the high-frequency content that the downsampling operator suppresses and that the correction must restore.
For XCT the error is larger and spatially structured, concentrating along streak-like features and near the boundary of the reconstruction region, consistent with the higher Res-NRMSE in Table~\ref{tab:results}.

\begin{figure*}[t]
    \centering
    \includegraphics[width=0.93\linewidth]{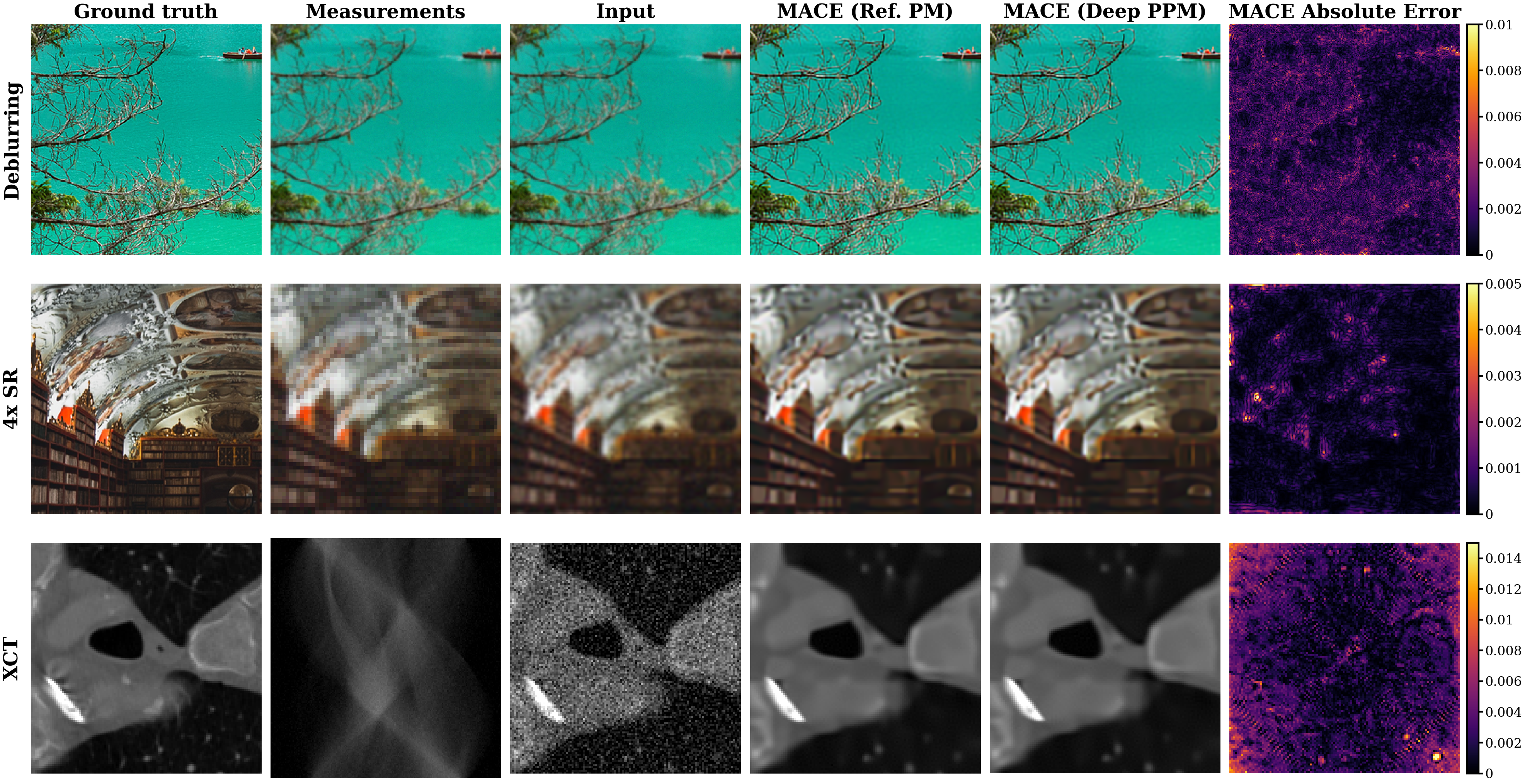}
    \caption{Comparison of MACE reconstructions using the reference proximal map and the deep PPM as the data-fitting agent, with a BM3D prior after 10 iterations. From left to right, each row displays the ground truth, measurement $y$, MACE initialization (blurry image, bicubically upsampled image, or FBP reconstruction), MACE reconstruction with the reference proximal map $x_{\mathrm{ref}}$, MACE reconstruction with the deep PPM $x_{\mathrm{PPM}}$, and the absolute difference $|x_{\mathrm{ref}} - x_{\mathrm{PPM}}|$ (mean over color channels). The reference is the analytic proximal map for deblurring and SR, and the CG solution for XCT. The two reconstructions are visually indistinguishable in all three cases.}
    \vspace{-5mm}
    \label{fig:mace}
\end{figure*}

Figure~\ref{fig:mace} compares MACE reconstructions using the deep PPM and the reference proximal map.
Replacing the reference proximal map with the deep PPM produces no visible change in the reconstructions, and the absolute error maps are very small compared to the $[0,1]$ pixel intensity scale.
The maps are mostly unstructured for deblurring and SR, but they contain faint residual streak structure for XCT.

\begin{figure*}[!t]
    \centering
    \includegraphics[width=\linewidth]{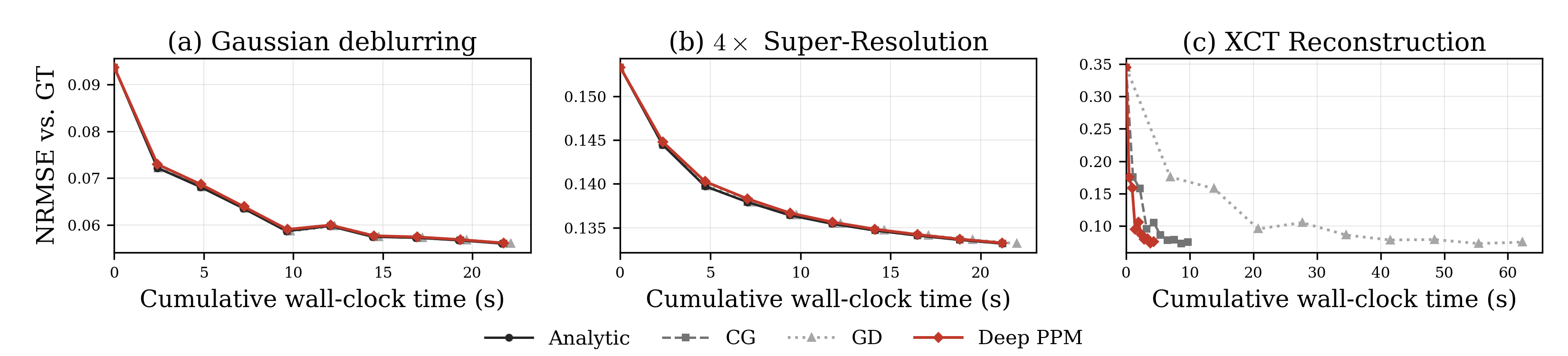}
    \caption{NRMSE with respect to ground truth versus cumulative wall-clock time over 10 MACE iterations, averaged over test images, using each proximal map variant as the data-fitting agent. For deblurring and SR the four trajectories coincide, and all variants finish in the same time since the BM3D prior dominates the per-iteration cost. For XCT the deep PPM follows the CG trajectory to the same equilibrium in 4.3 s, compared with 9.7 s for CG and 62 s for GD.}
    \vspace{-5mm}
    \label{fig:mace_conv}
\end{figure*}

Figure~\ref{fig:mace_conv} plots NRMSE of MACE iterates using each proximal map variant, plotted against wall-clock time.
For deblurring and SR, the deep PPM curve is indistinguishable from the analytic, CG, and GD curves at every iteration, and all four finish at the same time since BM3D accounts for nearly all of the per-iteration cost.
For XCT the deep PPM tracks the CG trajectory at every iteration and reaches the same NRMSE in 4.3\,s, compared with 9.7\,s for CG and 62\,s for GD, whose inner loops dominate the run time.

\section{Conclusion}
In this paper, we introduced the PPM, a principled reformulation of the data-fitting proximal map as the MMSE estimate of the PnP virtual model that enables training the deep PPM using only sampled Gaussian noise and black-box access to $A$ and $A^T$.
The PPM is provably equivalent to the proximal map when $A$ is linear, and the deep PPM closely matches the analytic proximal maps for deblurring and SR as well as the CG proximal map for XCT.
Used as the data-fitting agent in MACE, the deep PPM reconstruction matches the reference reconstruction to within 1\% NRMSE on average across all three operators, and for XCT it replaces the inner CG loop with a network evaluation that is $18\times$ faster.
Since the formulation is well-defined for any $A$, the PPM extends naturally to nonlinear forward models, where data-fitting proximal maps are most expensive; training and verifying the deep PPM for non-linear operators is the subject of future work.

\bibliographystyle{IEEEtran}
\bibliography{ref}

\end{document}